\documentclass[11pt,letterpaper]{article}

\usepackage{amsmath,amsfonts,amsthm,amssymb,stmaryrd,graphicx,relsize,bm}  % 9.20 add graphicx,resize
\theoremstyle{plain}

\numberwithin{equation}{section}

\newtheorem{thm}{Theorem}[section]
\newtheorem{lem}[thm]{Lemma}
\newtheorem{cor}[thm]{Corollary}

\newenvironment{exam}[1]
{\begin{flushleft}\textbf{Example #1}.\enspace}
{\end{flushleft}}

\allowdisplaybreaks  % introduced 7.15.15

\newcounter{cond}

\newcommand{\complex}{{\mathbb C}}
\newcommand{\real}{{\mathbb R}}

\newcommand{\tbullet}{\mathrel{\raise .4ex\hbox{\tiny$\bullet$}}} % 5.8.20 THIS for larger cdot as times
\newcommand{\deltasub}{{_\Delta}} % 8.21
\newcommand{\zeroone}{{0\,--1\ }} % 8.21
\newcommand{\onezero}{{1--\,0\ }} % 8.21

\newcommand{\rmtr}{\mathrm{tr\,}}

\newcommand{\rmprob}{\mathrm{Prob\,}}
\newcommand{\rmdiag}{\mathrm{diag\,}}
\newcommand{\rmre}{\mathrm{Re\,}}

\newcommand{\bscript}{\mathcal{B}}

\newcommand{\escript}{\mathcal{E}}
\newcommand{\fscript}{\mathcal{F}}
\newcommand{\gscript}{\mathcal{G}}
\newcommand{\iscript}{\mathcal{I}}
\newcommand{\jscript}{\mathcal{J}}

\newcommand{\lscript}{\mathcal{L}}

\newcommand{\sscript}{\mathcal{S}}
\newcommand{\tscript}{\mathcal{T}}

\newcommand{\iscripthat}{\widehat{\iscript}}
\newcommand{\jscripthat}{\widehat{\jscript}}
\newcommand{\vtilde}{\widetilde{V}}

\newcommand{\ab}[1]{\left|#1\right|}
\newcommand{\doubleab}[1]{\left\|#1\right\|}
\newcommand{\brac}[1]{\left\{#1\right\}}
\newcommand{\paren}[1]{\left(#1\right)}
\newcommand{\sqbrac}[1]{\left[#1\right]}
\newcommand{\elbows}[1]{{\left\langle#1\right\rangle}}
\newcommand{\ket}[1]{{\left|#1\right>}}
\newcommand{\bra}[1]{{\left<#1\right|}}
\newcommand{\sqparen}[1]{{\left[#1\right)}}

\begin{document}

\title{COARSE-GRAINING OF OBSERVABLES}
\author{Stan Gudder\\ Department of Mathematics\\
University of Denver\\ Denver, Colorado 80208\\
sgudder@du.edu}
\date{}
\maketitle

\begin{abstract}
We first define the coarse-graining of probability measures in terms of stochastic kernels. We define when a probability measure is part of another probability measure and say that two probability measures coexist if they are both parts of a single probability measure. We then show that any two probability measures coexist. We extend these concepts to observables and instruments and mention that two observables need not coexist. We define the discretization of an observable as a special case of coarse-graining and show that these have \zeroone stochastic kernels. We next consider finite observables and instruments and show that in these cases, stochastic kernels are replaced by stochastic matrices. We also show that coarse-graining is the same as post-processing in this finite case. We then consider sequential products of observables and discuss the sequential product of a post-processed observable with another observable. We briefly discuss SIC observables and the example of qubit observables.
\end{abstract}

\section{Coarse-Graining of Measures}  % Section 1
We denote the set of probability measures on a measurable space $(\Omega ,\fscript )$ by $\rmprob (\Omega ,\fscript )$. Let $(\Omega _1,\fscript _1)$,
$(\Omega _2,\fscript _2)$ be measurable spaces. A map $v\colon\Omega _1\times\fscript _2\to\sqbrac{0,1}$ that satisfies $x\mapsto v(x,\Delta )$ is measurable for all $\Delta\in\fscript _2$ and $v(x,\tbullet )\in\rmprob (\Omega _2,\fscript _2)$ for all $x\in\Omega _1$ is called a
\textit{stochastic kernel} \cite{hz12}. If $v$ is a stochastic kernel, define
\begin{equation*}
V\colon\rmprob (\Omega _1,\fscript _1)\to\rmprob (\Omega _2,\fscript _2)
\end{equation*}
by $\sqbrac{V(\mu )}(\Delta )=\int v(x,\Delta )\mu (dx)$. We call $v$ the stochastic kernel for $V$ and we say that $V(\mu )$ is a
\textit{coarse-graining} of $\mu$. We think of $V(\mu )$ as an imprecise version of $\mu\in\rmprob (\Omega _1,\fscript _1)$ on $(\Omega _2,\fscript _2)$. Notice that $V$ is an affine map because if $0\le\lambda _i\le 1$, $\sum\lambda _i=1$, then
\begin{align*}
\sqbrac{V\paren{\sum\lambda _i\mu _i}}(\Delta )&=\int v(x,\Delta )\paren{\sum\lambda _i\mu _i}(dx)=\sum\lambda _i\int v(x,\Delta )\mu _i (dx)\\
   &=\sum\lambda _i\sqbrac{V(\mu _i)}(\Delta )
\end{align*}
for all $\Delta\in\fscript _2$. Moreover, if $f\colon\Omega _2\to\real$ is measurable, then
\begin{align*}
\int _{\Omega _2}f(y)\sqbrac{V(\mu )}(dy)&\!=\!\int _{\Omega _2}\!f(y)\int _{\Omega _1}v(x,dy)\mu (dx)
    \!=\!\int _{\Omega _2}\int _{\Omega _1}\!f(y)v(x,dy)\mu (dx)\\
    &=\int _{\Omega _1}\sqbrac{\int _{\Omega _2}f(y)v(x,dy)}\mu (dx)
\end{align*}

\begin{exam}{1}  % Example 1
The map $v(x,\Delta )=\chi _{\deltasub}(x)$ is a stochastic kernel from $(\Omega ,\fscript )$ to $(\Omega ,\fscript )$. The corresponding coarse-graining map $V\colon\rmprob (\Omega ,\fscript )\to\rmprob (\Omega ,\fscript )$ satisfies
\begin{equation*}
V(\mu )(\Delta )=\int v(x,\Delta )\mu (dx)=\int\chi _{\deltasub}(x)\mu (dx)=\int _\Delta\mu (dx)=\mu (\Delta )
\end{equation*}
for all $\Delta\in\fscript$. Hence, $V(\mu )=\mu$ so $V$ is the identity map.\hfill\qedsymbol
\end{exam}

\begin{exam}{2}  % Example 2
Let $\nu\in\rmprob (\Omega _2,\fscript _2)$ and let $v\colon\Omega _1\times\fscript _2\to\sqbrac{0,1}$ be defined by $v(x,\Delta )=\nu (\Delta )$ for all
$x\in\Omega _1$, $\Delta\in\fscript _2$. Then $v$ is a stochastic kernel and the corresponding coarse-graining map is
\begin{equation*}
V(\mu )(\Delta )=\int v(x,\Delta )\mu (dx)=\int\nu (\Delta )\mu (dx)=\nu (\Delta )
\end{equation*}
Hence, $V$ is the constant map $V(\mu )=\nu$.\hfill\qedsymbol
\end{exam}

We define the \textit{Dirac measure at} $x$ on $(\Omega ,\fscript )$ by $\delta _x$ where $\delta _x(\Delta )=1$ if and only if $x\in\Delta$.

\begin{lem}    % Lemma 1.1
\label{lem11}
{\rm{(a)}}\enspace If $v\colon\Omega _1\times\fscript _2\to\sqbrac{0,1}$ is a stochastic kernel for
$V\colon\rmprob (\Omega _1,\fscript _1)\to(\Omega _2,\fscript _2)$ then $v(x,\Delta )=V(\delta _x)(\Delta )$.
{\rm{(b)}}\enspace If $V\colon\rmprob (\Omega _1,\fscript _1)\to\rmprob (\Omega _2,\fscript _2)$ has a stochastic kernel $v$, then $v$ is unique.
\end{lem}
\begin{proof}
(a)\enspace If $v$ is a stochastic kernel for $V$, then
\begin{equation*} 
V(\delta _x)(\Delta )=\int v(y,\Delta )\delta _x(dy)=v(x,\Delta )
\end{equation*}
(b)\enspace follows from (a).
\end{proof}

It can be shown that an arbitrary affine map $V\colon\rmprob (\Omega _1,\fscript _1)\to\rmprob (\Omega _2,\fscript _2)$ need not have a stochastic kernel and hence need not be a coarse-graining. One way to accomplish this is to construct such a map $V$ where $x\mapsto V(\delta _x)(\Delta )$ is not measurable for some $\Delta\in\fscript _2$. We leave the details of this to the reader. Then $V$ does not have a stochastic kernel $v$ because if it did, then by Lemma~\ref{lem11}(a), $v(x,\Delta )=V(\delta _x)(\Delta )$ so $x\mapsto v(x,\Delta )$ is not measurable for some $\Delta\in\fscript _2$ which is a contradiction.

Let $(\Omega _j\fscript _j)$, $j=1,2,3$, be measurable spaces and let $v\colon\Omega _1\times\fscript _2\to\sqbrac{0,1}$,
$u\colon\Omega _2\times\fscript _3\to\sqbrac{0,1}$ be stochastic kernels. Define $u\circ v\colon\Omega _1\times\fscript _3\to\sqbrac{0,1}$ by
$u\circ v(x,\Delta )=\int _{\Omega _2}u(y,\Delta )v(x,dy)$. Then $u\circ v$ is a stochastic kernel.

\begin{lem}    % Lemma 1.2
\label{lem12}
Let $V\colon\rmprob (\Omega _1,\fscript _1)\to\rmprob (\Omega _2,\fscript _2)$ and
$U\colon\rmprob (\Omega _2,\fscript _2)\to\rmprob (\Omega _3,\fscript _3)$ be coarse-grainings with corresponding stochastic kernels $v,u$. Then their composition $U\circ V\colon\rmprob (\Omega _1,\fscript _1)\to\rmprob (\Omega _3,\fscript _3)$ has stochastic kernel $u\circ v$.
\end{lem}
\begin{proof}
For all $\mu\in\rmprob (\Omega _1,\fscript _1)$, $\Delta\in\fscript _3$ we have that
\begin{align*}
\sqbrac{U\circ V(\mu )}(\Delta )&=\int _{\Omega _2}u(y,\Delta )V(\mu )(dy)=\int _{\Omega _2}\int _{\Omega _1}u(y,\Delta )v(x,dy)\mu (dx)\\
   &=\int _{\Omega _1}\int _{\Omega _2}u(y,\Delta )v(x,dy)\mu (dx)=\int _{\Omega _1}u\circ v (x,\Delta )\mu (dx)
\end{align*}
Hence, the stochastic kernel for $U\circ V$ is $u\circ v$.
\end{proof}

We say that $\mu _2\in\rmprob (\Omega _2,\fscript _2)$ is \textit{part} of $\mu _1\in\rmprob (\Omega _1,\fscript _1)$ if there exists a measurable function $f\colon\Omega _1\to\Omega _2$ such that $\mu _2(\Delta )=\mu _1\sqbrac{f^{-1}(\Delta )}$ for all $\Delta\in\fscript _2$. Define
$V_f\colon\rmprob (\Omega _1,\fscript _1)\to\rmprob (\Omega _2,\fscript _2)$ by $(V_f\mu )(\Delta )=\mu\sqbrac{f^{-1}(\Delta )}$. Thus, $\mu _2$ is part of $\mu _1$, if and only if $\mu _2=V_f(\mu _1)$ for a measurable function $f\colon\Omega _1\to\Omega _2$. Notice that $V_f$ is affine because
\begin{equation*} 
\sqbrac{V_f\paren{\sum\lambda _i\mu _i)}(\Delta )}=\sum\lambda _i\mu _i\sqbrac{(V_f\mu _i)(\Delta )}=\paren{\sum\lambda _iV_f\mu _i)(\Delta )}
\end{equation*}
and hence, $V_f\paren{\sum\lambda _iu_i}=\sum\lambda _iV_f(\mu _i)$.

\begin{lem}    % Lemma 1.3
\label{lem13}
A map $v\colon\Omega _1\times\fscript _2\to\sqbrac{0,1}$ is the stochastic kernel for $V_f$ if and only if $v(x,\Delta )=\chi _{f^{-1}(\Delta )}(x)$ for all
$x\in\Omega _1$, $\Delta\in\fscript _2$.
\end{lem}
\begin{proof}
If $v(x,\Delta )=\chi _{f^{-1}(\Delta )}$, then $v$ is a stochastic kernel and
\begin{align*}
\int v(x,\Delta )\mu (dx)&=\int\chi _{f^{-1}(\Delta )}(x)\mu (dx)=\int _{f^{-1}(\Delta )}\mu (dx)\\
   &=\mu\sqbrac{f^{-1}(\Delta )}=(V_f\mu )(\Delta )
\end{align*}
for all $x\in\Omega _1$, $\Delta\in\fscript _2$. Hence, $v$ is a stochastic kernel for $V_f$. Since stochastic kernels are unique, the converse holds.
\end{proof}

\begin{lem}    % Lemma 1.4
\label{lem14}
If $f\colon\Omega _1\to\Omega _2$ and $g\colon\Omega _2\to\Omega _3$ are measurable, $V_g\circ V_f=V_{g\circ f}$ and the stochastic kernel for $V_{g\circ f}$ is $\omega (x,\Delta )=\chi _{f^{-1}\paren{g^{-1}(\Delta )}}(x)$.
\end{lem}
\begin{proof}
For all $\mu\in\rmprob (\Omega _1,\fscript _1)$ and $\Delta\in\fscript _3$ we have that
\begin{align*}
(V_g\circ V_f)\mu (\Delta )&=V_f\sqbrac{\mu\paren{g^{-1}(\Delta )}}=\mu\sqbrac{f^{-1}\paren{g^{-1}(\Delta )}}=\mu\sqbrac{(g\circ f)^{-1}(\Delta )}\\
   &=(V_{g\circ f}\mu )(\Delta )
\end{align*}
Hence, $V_g\circ V_f=V_{g\circ f}$. It follows from Lemma~\ref{lem13} that the stochastic kernel for $V_{g\circ f}$ is
\begin{equation*}
w(x,\Delta )=\chi _{(g\circ f)^{-1})\Delta )}(x)=\chi _{f^{-1}\paren{g^{-1}(\Delta )}}(x)\qedhere
\end{equation*}
\end{proof}

We say that two probability measures \textit{coexist}, if they are both parts of another probability measure.

\begin{lem}    % Lemma 1.5
\label{lem15}
If $\mu _1\in\rmprob (\Omega _1,\fscript _1)$, $\mu _2\in\rmprob (\Omega _2\fscript _2)$ then $\mu _1,\mu _2$ coexist.
\end{lem}
\begin{proof}
Define $\mu\in\rmprob (\Omega _1\times\Omega _2,\fscript _1\times\fscript _2)$ by $\mu =\mu _1\times\mu _2$ and define
$f\colon\Omega _1\times\Omega _2\to\Omega _1$ by $f(x,y)=x$, $g\colon\Omega _1\times\Omega _2\to\Omega _2$ by $g(x,y)=y$. Then $f$ and $g$ are measurable and if $\Delta _1\in\fscript _1$ we obtain
\begin{equation*}
\mu\sqbrac{f^{-1}(\Delta _1)}=\mu (\Delta _1\times\Omega _2)=\mu _1(\Delta _1)\mu _2(\Omega _2)=\mu _1(\Delta _1)
\end{equation*}
Hence, $\mu _1$ is a part of $\mu$. Similarly, if $\Delta _2\in\fscript _2$, then $\mu\sqbrac{g^{-1}(\Delta _2)}=\mu _2(\Delta _2)$ so $\mu _2$ is a part of $\mu$.
\end{proof}

Let $(\Omega ,\fscript )$ be a measurable space and let $(\Omega _1,2^{\Omega _1})$ be a finite measurable space with
$\Omega _1=\brac{1,2,\ldots ,n}$. Let $B_1,B_2,\ldots ,B_n$ be a \textit{measurable partition} of $\Omega$. That is
$B_i\cap B_j=\emptyset$ for $i\ne j$ and $\cup B_i=\Omega$. Define
\begin{equation*}
V\colon\rmprob (\Omega ,\fscript )\to\rmprob (\Omega _1,2^{\Omega _1})
\end{equation*}
by $V(\mu )(\Delta )=\sum\brac{\mu (B_i)\colon i\in\Delta}$. Then $V$ is affine because
\begin{align*}
V\paren{\sum\lambda _j\mu _j}(\Delta )&=\sum _{i\in\Delta}\sqbrac{\sum _j\lambda _j\mu _j(B_i)}=\sum _j\lambda _j\sum _{i\in\Delta}\mu _j(B_i)\\
   &=\sum _j\lambda _jV(\mu _j)(\Delta )
\end{align*}
so that $V\paren{\sum\lambda _j\mu _j}=\sum\lambda _jV(\mu _j)$. We call $V$ a \textit{discretization map} and $V(\mu )$ a
\textit{discretization} of $\mu$. A stochastic kernel $v(x,\Delta )$ is called a \textrm{\zeroone} \textit{stochastic kernel} if $v(x,\Delta )=0\hbox{ or }1$ for all $x,\Delta$.

\begin{thm}    % Theorem 1.6
\label{thm16}
An affine map $V\colon\rmprob (\Omega ,\fscript )\to\rmprob (\Omega _1,2^{\Omega _1})$ is a discretization if and only if $V$ has a \zeroone stochastic kernel.
\end{thm}
\begin{proof}
Suppose $V$ is a discretization and $V$ has stochastic kernel $v(x,\Delta )$. Then by Lemma~\ref{lem11} we obtain for $j=1,2,\ldots ,n$ that
\begin{equation*}
v(x,\brac{j})=V(\delta _x)(\brac{j})=\sum\brac{\delta _x(B_i)\colon i\in\brac{j}}=\delta _x(B_j)
\end{equation*}
Hence,
\begin{align*}
v(x,\Delta )&=V(\delta _x)(\Delta )\sum _{j\in\Delta}V(\delta _x)(\brac{j})=\sum _{j\in\Delta}\delta _x(B_j)\\
   &=\delta _x\paren{\bigcup _{j\in\Delta}B_j}=\chi _{_{\!\bigcup\limits _{j\in\Delta}\!\!B_j}}(x)
\end{align*}
for all $x\in\Omega$, $\Delta\in 2^{\Omega _1}$. To show that $v(x,\Delta )$ is actually the stochastic kernel for $V$ we have that
\begin{align*}
\int v(x,\Delta )\mu (dx)&=\int \chi _{_{\!\bigcup\limits _{j\in\Delta}\!\!B_j}}(x)\mu (dx)=\int _{\bigcup\limits _{j\in\Delta}\!\!B_j}\mu (dx)\\
   &=\sum _{j\in\Delta}\mu (B_j)=V(\mu )(\Delta )
\end{align*}
Of course, $v(x,\Delta )$ is a \zeroone stochastic kernel. Conversely, suppose $v(x,\Delta )$ is a \zeroone stochastic kernel for
\begin{equation*}
V\colon\rmprob (\Omega ,\fscript )\to\rmprob (\Omega _1,2^{\Omega _1})
\end{equation*}
Then $v(x,\Delta )=\sum\limits _{j\in\Delta}v(x,\brac{j})$ for all $x\in\Omega$, $\Delta\in 2^{\Omega _1}$. Let $B_i$, $i=1,2,\ldots ,n$, be the measurable sets
\begin{equation*}
B_i\brac{x\in\Omega\colon v(x,\brac{i})=1}
\end{equation*}
If $x\in B_i\cap B_j$ for $i\ne j$, then $v(x,\brac{i})=v(x,\brac{j})=1$ and $v(x,\brac{i,j})=2$ which is a contradiction. Hence, $B_i\cap B_j=\emptyset$ for $i\ne j$. If $x\in\Omega$ and $v(x,\brac{i})=0$ for all $i=1,2,\ldots ,n$, then
\begin{equation*}
v(x,\Omega _1)=\sum _{i\in\Omega _1}v(x,\brac{i})=0
\end{equation*}
which is a contradiction. Hence, there exists an $i$ such that $v(x,\brac{i})=1$ so $\cup B_i=\Omega _1$. We conclude that $\brac{B_i}$ is a measurable partition of $\Omega$. Since
\begin{equation*}
(V\mu )(\brac{i})=\int v(x,\brac{i})\mu (dx)=\int _{B_i}\mu (dx)=\mu (B_i)
\end{equation*}
we have for all $\Delta\in 2^{\Omega _1}$ that
\begin{equation*}
(V\mu )(\Delta )=\sum _{i\in\Delta}(V\mu )(\brac{i})=\sum _{i\in\Delta}\mu (B_i)
\end{equation*}
We conclude that $V$ is a discretization map.
\end{proof}

When we consider a finite measurable space $(\Omega ,\fscript )$ we always assume that $\fscript =2^\Omega$ so $\fscript$ need not be specified. For $\Omega =\brac{x_1,x_2,\ldots ,x_n}$ we identify a $\mu\in\rmprob (\Omega )$ with the column vector with entries
$\mu (x_1),\mu (x_2),\ldots ,\mu (x_n)$ where we write $\mu \paren{\brac{x_i}}=\mu (x_i)$, $i=1,2,\ldots ,n$. An $m\times n$ matrix
$M=\sqbrac{m_{ij}}$ is a \textit{stochastic matrix} if $0\le m_{ij}\le 1$ and $\sum\limits _{j=1}^mm_{ij}=1$ for all $i=1,2,\ldots ,n$. In this finite case, the stochastic kernels are replaced by stochastic matrices. This is because, in the finite case, if $v(x,\Delta )$ is a stochastic kernel, then
$v(x_i,\brac{y_j})$ is a stochastic matrix and conversely, if $\sqbrac{m_{ij}}$ is a stochastic matrix, then 
\begin{equation*}
v(x_i,\Delta )=\sum\brac{m_{ij}\colon y_j\in\Delta}
\end{equation*}
is a stochastic kernel.

\begin{thm}    % Theorem 1.7
\label{thm17}
Let $\Omega _1=\brac{x_1,x_2,\ldots ,x_n}$, $\Omega _2=\brac{y_1,y_2,\ldots ,y_m}$ and let
$V\colon\rmprob (\Omega _1)\to\rmprob (\Omega _2)$ be affine. Then there exists a unique $m\times n$ stochastic matrix $\vtilde$ such that for every $\nu\in\rmprob (\Omega _1)$ we have $V(\nu )=\vtilde\nu$. Conversely, if $M$ is an $m\times n$ stochastic matrix, then there exists an affine map $V\colon\rmprob (\Omega _1)\to\rmprob (\Omega _2)$ such that $\vtilde =M$.
\end{thm}
\begin{proof}
Let $V\colon\rmprob (\Omega _1)\to\rmprob (\Omega _2)$ be affine. Since every element of $\rmprob (\Omega _2)$ is a convex combination of
$\delta _{y_j}$, $j=1,2,\ldots ,m$ we have that
\begin{equation*}
V(\delta _{x_i})=\sum _{j=1}^m\mu _{ij}\delta _{y_j}
\end{equation*}
where $0\le\mu _{ij}\le 1$ and $\sum _{j=1}^m\mu _{ij}=1$, $i=1,2,\ldots ,n$. We conclude that $\vtilde =\sqbrac{\mu _{ij}}$ is an $m\times n$ stochastic matrix and $\mu _{ij}=\sqbrac{V(\delta _{x_i})}(y_j)$. Letting $\nu\in\rmprob (\Omega _1)$ we obtain $\nu =\sum\limits _{i=1}^n\nu _i\delta _{x_i}$ where
$\nu _i=\nu (x_)$, $i=1,2,\ldots ,n$. Since $0\le\nu _i\le 1$, $\sum\limits _{i=1}^n\nu _i=1$ and $V$ is affine, we conclude that
\begin{equation*}
V(\nu )=V\paren{\sum _{i=1}^n\nu _i\delta _{x_i}}=\sum _{i=1}^n\nu _iV(\delta _{x_i})=\sum _{i=1}^n\nu _i\sum _{j=1}^m\mu _{ij}\delta _{y_j}
   =\sum _{i,j}\mu _{ij}\nu (x_i)\delta _{y_j}
\end{equation*}
It follows that
\begin{equation*}
V(\nu )=\begin{bmatrix}V(\nu )(y_1)\\V(\nu )(y_2)\\\vdots\\V(\nu )(y_m)\end{bmatrix}
   =\begin{bmatrix}\sum\mu _{i1}\nu (x_i)\\\sum\mu _{i2}\nu (x_i)\\\vdots\\\sum\mu _{im}\nu (x_i)\end{bmatrix}
  =\vtilde\begin{bmatrix}\nu (x_1)\\\nu (x_2)\\\vdots\\\nu (x_n)\end{bmatrix} =\vtilde\nu
\end{equation*}
To show that $\vtilde$ is unique, suppose $V(\nu )=M\nu$ where $M=\sqbrac{M_{ij}}$ is an $m\times n$ matrix. We then obtain 
\begin{equation*}
M_{ij}=\elbows{\delta _{y_j},M\delta _{x_i}}=\elbows{\delta _{y_j},V(\delta _{x_i}}=\elbows{\delta _{y_j},\sum _{k=1}^m\mu _{ik}\delta _{y_k}}
   =\mu _{ij}=\vtilde _{ij}
\end{equation*}
Conversely, let $M=\sqbrac{M_{ij}}$ be an $m\times n$ stochastic matrix. Define $V\colon\rmprob (\Omega _1)\to\rmprob (\Omega _2)$ by
$V(\delta _{xi})=\sum\limits _{j=1}^m\mu _{ij}\delta _{y_j}$, $i=1,2,\ldots ,n$ and extend $V$ affinely to all of $\rmprob (\Omega _1)$. By our previous work, $\vtilde =M$.
\end{proof}

We conclude that in the finite case, every affine map $V\colon\rmprob (\Omega _1)\to\rmprob (\Omega _2)$ is a coarse-graining and is implemented by a unique stochastic matrix $\vtilde$. We then identify $V$ and $\vtilde$.

\section{Observables and Instruments}  % Section 2
In this section we employ our previous work to study coarse-graining of observables and instruments. Let $H$ be a complex Hilbert space that represents a quantum system $S$. We denote the set of bounded linear operators on $H$ by $\lscript (H)$. For $A,B\in\lscript (H)$, we write $A\le B$ if $\elbows{\phi ,A\phi}\le\elbows{\phi ,B\phi}$ for all $\phi\in H$. An operator $E\in\lscript (H)$ is an \textit{effect} if $0\le E\le I$ where $0,I$ are the zero and identity operators respectively. We denote the set of effects by $\escript (H)$ and interpret an $E\in\escript (H)$ as a \onezero (true-false) measurement \cite{bgl95,hz12,nc00}. If $(\Omega _A,\fscript )$ is a measurable space, an \textit{observable} with \textit{outcome space} $\Omega _A$ is an effect-valued measure $A\colon\fscript\to\escript (H)$ \cite{bgl95,hz12,nc00}. That is, $A(\cup\Delta _i)=\sum A(\Delta _i)$ when
$\Delta _i\cap\Delta _j=\emptyset$, $i\ne j$, and $A(\Omega )=I$. We interpret $A(\Delta )$ as the effect that occurs when a measurement of $A$ results in an outcome in $\Delta$. A \textit{state} for $S$ is an effect $\rho\in\escript (H)$ that satisfies $\rmtr (\rho )=1$. We denote the set of states on $H$ by $\sscript (H)$. If $\rho\in\sscript (H)$, $E\in\escript (H)$ we interpret $\rmtr (\rho E)$ as the probability that $E$ occurs (is true) when $S$ is in the state $\rho$. If $A$ is an observable, its statistics in the state $\rho$ is given by the \textit{distribution}
\begin{equation*}
\Phi _\rho ^A(\Delta )=\rmtr\sqbrac{\rho A(\Delta )}
\end{equation*}
for all $\Delta\in\fscript$. Of course, $\Phi _\rho ^A\in\rmprob (\Omega ,\fscript )$ for all $\rho\in\sscript (H)$ \cite{bgl95,hz12,nc00}.

We now discuss a method for constructing stochastic kernels from observables. Let $(\Omega ,\fscript )$, $(\Omega _A,\gscript )$ be measurable spaces, $\brac{\alpha _x\colon x\in\Omega}\subseteq\sscript (H)$ a collection of states and $A$ an observable with outcome space $\Omega _A$. We say that $(\alpha ,A)$ is \textit{measurable} if $x\mapsto\Phi _{\alpha _x}^A(\Delta )$ is measurable for all $\Delta\in\gscript$. If $(\alpha ,A)$ is measurable, we define the stochastic kernel
\begin{equation}                % equation (2.1)
\label{eq21}
v(x,\Delta )=\rmtr\sqbrac{\alpha _xA(\Delta )}=\Phi _{\alpha _x}^A(\Delta )
\end{equation}
with the corresponding coarse-graining
\begin{equation}                % equation (2.2)
\label{eq22}
V_{(\alpha ,A)}(\mu )(\Delta )=\int v(x,\Delta )\mu (dx)=\int\rmtr\sqbrac{\alpha _xA(\Delta )}\mu (dx)
   =\int\phi _{\alpha _x}^A(\Delta )\mu (dx)
\end{equation}
If $\alpha _x$ are pure states $\alpha _x=\ket{\phi _x}\bra{\phi _x}$, $\phi _x\in H$, then \eqref{eq21} and \eqref{eq22} become
\begin{align}               % equation (2.3) & equation (2.4)   
\label{eq23}
v(x,\Delta )&=\elbows{A(\Delta )\phi _x,\phi _x}
\intertext{and}
\label{eq24}
V_{(\alpha ,A)}(\mu )(\Delta )&=\int\elbows{A(\Delta )\phi _x,\phi _x}\mu (dx)
\end{align}
We interpret \eqref{eq21} as the probability that a measurement of $A$ results in an outcome in $\Delta$ when $S$ is in the state $\alpha _x$.

\begin{exam}{3}  % Example 3
Let $\Omega =\brac{1,2,\ldots ,n}$ be a finite measurable space. We show that any stochastic matrix $M=\sqbrac{\mu _{ij}}$, $i,j=1,2,\ldots ,n$ can be written in the form of the previous paragraph. Let $H$ be a complex Hilbert space with dimension $n$ and let $\brac{\phi _i\colon i=1,2,\ldots ,n}$ be an orthonormal basis for $H$. Let $A$ be the observable with outcome space $\Omega$ satisfying
\begin{equation*}
A\paren{\brac{j}}=\rmdiag\sqbrac{\mu _{1j},\mu _{2j},\ldots ,\mu _{nj}}
\end{equation*}
Letting $\alpha _i$ be the pure state $\alpha _i=\ket{\phi _i}\bra{\phi _i}$, $i=1,2,\ldots ,n$, we obtain
\begin{equation*}
\elbows{A\paren{\brac{j}}\phi _i,\phi _i}=\mu _{ij}
\end{equation*}
This is essentially \eqref{eq23}.\hfill\qedsymbol
\end{exam}

We now give an application of the previous structure to the study of the dynamics of the system $S$. Suppose the dynamics of $S$ is described by the strongly continuous unitary group $e^{-itK}$, $t\in\sqparen{0,\infty}$, where $K$ is the Hamiltonian for $S$. If $\phi _0\in H$ is the initial state, then
$\phi _t=e^{-itK}\phi _0$ is the state at time $t\in\sqparen{0,\infty}$. We can consider $\phi _t$ as a collection of states indexed by the points of the measurable space, $\paren{\sqparen{0,\infty},\bscript\paren{\sqparen{0,\infty}}}$. Let $A$ be an observable with outcome space $(\Omega _A,\fscript )$. Since $t\mapsto\phi _t$ is continuous we have that
\begin{equation}                % equation (2.5)
\label{eq25}
t\mapsto\Phi _{\phi _t}^A(\Delta )=\elbows{A(\Delta )\phi _t,\phi _t}
\end{equation}
is continuous for all $\Delta\in\fscript$. It follows that $(\phi _t,A)$ is measurable. We conclude that the map
$v\colon\sqparen{0,\infty}\times\fscript\to\sqbrac{0,1}$ given by $v(t,\Delta )=\elbows{A(\Delta )\phi _t,\phi _t}$ is a stochastic kernel called the
\textit{dynamical kernel} for $(\phi _t,A)$. We interpret $v(t,\Delta )$ as the probability that a measurement of $A$ at time $t$ results in an outcome in
$\Delta$. In terms of the dynamical group we have
\begin{equation}                % equation (2.6)
\label{eq26}
v(t,\Delta )=\elbows{A(\Delta )e^{-itK}\phi _0,e^{-itK}\phi _0}=\elbows{e^{itK}A(\Delta )e^{-itk}\phi _0,\phi _0}
\end{equation}
The observable $\Delta\mapsto e^{itK}A(\Delta )e^{-itK}$ which gives the time evolution of $A$ is the \textit{Heisenberg picture} of quantum mechanics while \eqref{eq25} gives the \textit{Schr\"odinger picture}. The corresponding coarse-graining map
\begin{align*}
&V_{(\phi ,A)}\colon\rmprob\paren{\sqparen{0,\infty},\bscript\paren{\sqparen{0,\infty}}}\to\rmprob (\Omega _A,\fscript )
\intertext{satisfies}
&V_{(\phi ,A)}(\mu )(\Delta )=\int\elbows{A(\Delta )\phi _t,\phi _t}\mu (dt)=\int\elbows{e^{itK}A(\Delta )e^{-itK}\phi _0,\phi _0}\mu (dt)
\end{align*}
For a particular time $t_0\in\sqparen{0,\infty}$ we have
\begin{equation*} 
V_{(\phi ,A)}(\delta _{t_0})(\Delta )=v(t_0,\Delta )=\elbows{e^{it_0K}A(\Delta )e^{-it_0K}\phi _0,\phi _0}
\end{equation*}

Let $A$ be an observable with outcome space $(\Omega _A,\fscript )$ and let $(\Omega ,\gscript )$ be a measurable space. If
$v\colon\Omega _A\times\gscript$ is a stochastic kernel, we define the observable $V\tbullet A$ with outcome space $\Omega$ by
\begin{equation*} 
V\tbullet A(\Delta )=\int v(x,\Delta )A(dx),\hbox{ for all }\Delta\in\gscript
\end{equation*}
We call $v$ the \textit{stochastic kernel} for $V$ and $V\tbullet A$ is a \textit{coarse-graining} of $A$ \cite{bgl95,hz12}. We see that $V\tbullet A(\Delta )$ is the unique effect satisfying
\begin{equation*} 
\rmtr\sqbrac{\rho V\tbullet A(\Delta )}=\int v(x,\Delta )\rmtr\sqbrac{\rho A(dx)}
\end{equation*}
for all $\rho\in\sscript (H)$. We now show that this idea extends to observables.

\begin{lem}    % Lemma 2.1
\label{lem21}
$V\tbullet A$ is the unique observable with distribution
\begin{equation*} 
\Phi _\rho ^{V\tbullet A}(\Delta )=V\sqbrac{\Phi _\rho ^A}(\Delta )
\end{equation*}
\end{lem}
\begin{proof}
For all $\rho\in\sscript (H)$, $\Delta\in\gscript$ we obtain
\begin{align*}
\Phi _\rho ^{V\tbullet A}(\Delta )&=\rmtr\sqbrac{\rho (V\tbullet A)(\Delta )}=\rmtr\sqbrac{\rho\int v(x,\Delta )A(dx)}
   =\int v(x,\Delta )\rmtr\sqbrac{\rho A(dx)}\\
   &=\int v(x,\Delta )\Phi _\rho ^A(dx)=V\sqbrac{\Phi _\rho ^A}(\Delta )
\end{align*}
The observable $V\tbullet A$ is unique because two observables on $H$ with the same distributions for every $\rho\in\sscript (H)$ are identical
\cite{bgl95,hz12,nc00}.
\end{proof}

If $A_i$, $i=1,2,\ldots ,n$, are observables on $H$ with the same outcome set and $0\le\lambda _i\le 1$, $\sum\lambda _i=1$, it is clear that
$\sum\lambda _iA_i$ is again an observable. Thus, such observables form a convex set. We conclude that $A\mapsto V\tbullet A$ is an affine map because
\begin{align*} 
V\tbullet\paren{\sum\lambda _iA_i}(\Delta )&=\int v(x,\Delta )\sum\lambda _iA(dx)=\sum\lambda _i\int v(x,\Delta )A_i(dx)\\
   &=\sum\lambda _i(V\tbullet A_i)(\Delta )
\end{align*}

Let $(\Omega ,\fscript )$, $(\Omega _A,\gscript )$ be measurable spaces and let $(\alpha ,A)$ be measurable with corresponding stochastic kernel
$v(x,\Delta )$ and coarse-graining $V_{(\alpha ,A)}$ given by \eqref{eq21} and \eqref{eq22}. If $B$ is an observable with outcome space $\Omega$ we obtain the following result.

\begin{lem}    % Lemma 2.2
\label{lem22}
{\rm{(a)}}\enspace For all $\Delta\in\gscript$ we have that
\begin{equation*}
(V_{(\alpha ,A)}\tbullet B)(\Delta )=\int\Phi _{\alpha _x}^A(\Delta )B(dx)
\end{equation*}
{\rm{(b)}}\enspace For all $\rho\in\sscript (H)$, $\Delta\in\gscript$, we have that
\begin{equation*}
\Phi _\rho^{V_{(\alpha ,A)\tbullet B}}(\Delta )=\int\Phi _{\alpha _x}^A(\Delta )\Phi _\rho ^B(dx)
\end{equation*}
\end{lem}
\begin{proof}
(a)\enspace Since $v(x,\Delta )=\Phi _{\alpha _x}^A(\Delta )$ for all $\Delta\in\gscript$, we obtain
\begin{equation*}
(V_{(\alpha ,A)}\tbullet B)(\Delta )=\int v(x,\Delta )B(dx)=\int\Phi _{\alpha _x}^A(\Delta )B(dx)
\end{equation*}
(b)\enspace For all $\rho\in\sscript (H)$, $\Delta\in\gscript$, applying (a) we obtain
\begin{align*}
\Phi _\rho ^{V_{(\alpha ,A)\tbullet B}}(\Delta )&=\rmtr\sqbrac{\rho (V_{(\alpha ,A)}\tbullet B)(\Delta )}
   =\rmtr\sqbrac{\rho\int\Phi _{\alpha _x}^A(\Delta )B(dx)}\\
   &=\int\Phi _{\alpha _x}^A(\Delta )\rmtr\sqbrac{\rho B(dx)}=\int\Phi _{\alpha _x}^A(\Delta )\Phi _\rho ^B(dx)\qedhere
\end{align*}
\end{proof}

An observable $B$ is \textit{part} of an observable $A$ if there exists a measurable surjection $f\colon\Omega _A\to\Omega _B$ such that
$B=V_f\tbullet A$ \cite{fhl18,gud120,gud220}.

\begin{lem}    % Lemma 2.3
\label{lem23}
Let $A,B$ be observables on $H$ with outcome spaces $(\Omega _A,\fscript _A)$, $(\Omega _B,\fscript _B)$ respectively. Then $B$ is part of $A$ if and only if there is a measurable surjection $f\colon\Omega _A\to\Omega _B$ such that $B(\Delta )=A\sqbrac{f^{-1}(\Delta )}$ for all
$\Delta\in\fscript _B$.
\end{lem}
\begin{proof}
If $B$ is a part of $A$, there exists a measurable surjection $f\colon\Omega _A\to\Omega _B$ such that $B=V_f\tbullet A$. If
$v(x,\Delta )=\chi _{f^{-1}(\Delta )}(x)$ is the corresponding stochastic kernel, then for $\Delta\in\fscript _B$ we obtain
\begin{align*}
B(\Delta )&=(V_f\tbullet A)(\Delta )=\int v(x,\Delta )A(dx)=\int\chi _{f^{-1}(\Delta )}(x)A(dx)\\
   &=\int _{f^{-1}(\Delta )}A(dx)=A\sqbrac{f^{-1}(\Delta )}
\end{align*}
Conversely, if $B(\Delta )=A\sqbrac{f^{-1}(\Delta )}$ for all $\Delta\in\fscript _B$, then letting $v(x,\Delta )=\chi _{f^{-1}(\Delta )}(x)$ we obtain
$B(\Delta )=(V_f\tbullet A)(\Delta )$ by reversing the previous argument. Hence, $B=V_f\tbullet A$ so $B$ is part of $A$.
\end{proof}

By Lemma~\ref{lem21} if $B$ is part of $A$ so that $B=V_f\tbullet A$, then $\Phi _\rho ^B=V_f(\Phi _\rho ^A)$ and hence $\Phi _\rho ^B$ is part of
$\Phi _\rho ^A$ for all $\rho\in\sscript (H)$. Two observables $B,C$ \textit{coexist} if there exists an observable $A$ such that $B$ and $C$ are part of $A$ \cite{bgl95,hz12,hrsz09,lah03}. It is well-known that unlike in Lemma~\ref{lem15}, two observables need not coexist \cite{bgl95,hz12,lah03}. Let $A$ be an observable with outcome space $(\Omega _A,\fscript )$. If $V$ is a discretization of $(\Omega _A,\fscript )$, we call $V\tbullet A$ a
\textit{discretization} of $A$ \cite{hz12}. If $v(x,\brac{i})=\chi _{B_i}(x)$ is the corresponding stochastic kernel we obtain
\begin{equation}                % equation (2.7)
\label{eq27}
(V\tbullet A)(\brac{i})=\int v(x,\brac{i})A(dx)=\int\chi _{_{B_i}}(x)A(dx)=\int _{B_i}A(dx)=A(B_i)
\end{equation}
Moreover,
\begin{equation*}
(V\tbullet A)(\Delta )=\sum _{i\in\Delta}(V\tbullet A)(\brac{i})=\sum\brac{A(B_i)\colon i\in\Delta}
\end{equation*}

\begin{lem}    % Lemma 2.4
\label{lem24}
If $V\tbullet A$ is a discretization of $A$, then $V\tbullet A$ is a part of $A$.
\end{lem}
\begin{proof}
Let $V\colon\rmprob (\Omega _A,\fscript )\to\rmprob (\Omega _1)$ where $\Omega _1=\brac{1,2,\ldots ,n}$ so the outcome space of $V\tbullet A$ is
$\Omega _1$. Let $v(x,\brac{i})=\chi _{_{B_i}}(x)$ be the corresponding stochastic kernel. Define $f\colon\Omega _A\to\Omega _1$ by $f(x)=i$ if
$x\in B_i$. Then by \eqref{eq27}
\begin{equation*}
(V\tbullet A)(\brac{i})=A(B_i)=A\sqbrac{f^{-1}(\brac{i})}
\end{equation*}
and it follows that for all $\Delta\subseteq\Omega _1$ we obtain
\begin{equation*}
(V\tbullet A)(\Delta )=\sum _{i\in\Delta}(V\tbullet A)(\brac{i})=A\sqbrac{f^{-1}(\Delta )}
\end{equation*}
Hence, $V\tbullet A$ is part of $A$.
\end{proof}

\begin{cor}    % Corollary 2.5
\label{cor25}
Any two discretizations  of an observable coexist.
\end{cor}

Let $\tscript (H)$ be the set of trace-class operators on $H$. An \textit{operation} on $H$ is a trace non-increasing, completely positive linear map
$T\colon\tscript (H)\to\tscript (H)$ \cite{bgl95,hz12,hrsz09,nc00}. If an operation $T$ preserves the trace, then $T$ is called a \textit{channel} on $H$. An \textit{instrument} on $H$ with \textit{outcome space} $\Omega _\iscript$ is an operation-valued measure $\iscript$ on $(\Omega _\iscript ,\fscript )$ such that $\iscript (\Omega _\iscript )$ is a channel. The statistics of an instrument $\iscript$ for a state $\rho\in\sscript (H)$ is given by its
\textit{distribution}
\begin{equation*}
\Phi _\rho ^\iscript (\Delta )=\rmtr\sqbrac{\iscript (\Delta )(\rho )}
\end{equation*}
for all $\Delta\in\fscript$. Of course, $\Phi _\rho ^\iscript$ is a probability measure on $(\Omega _\iscript ,\fscript )$. We say that an instrument $\iscript$ \textit{measures} an observable $A$ if $\Omega _A=\Omega _\iscript$ and for all $\rho\in\sscript (H)$ and $\Delta\in\fscript$ we have
\begin{equation*}
\Phi _\rho ^A(\Delta )=\rmtr\sqbrac{\rho A(\Delta )}=\rmtr\sqbrac{\iscript (\Delta )(\rho )}=\Phi _\rho ^\iscript (\Delta )
\end{equation*}
It can be shown that an instrument measures a unique observable, but an observable is measured by many instruments \cite{hz12}. If $\iscript$ measures $A$ we write $\iscripthat =A$. We think of $\iscript$ as an apparatus that can be employed to measure the observable $\iscripthat$ and conclude that there are many such apparatuses. Although $\iscript$ reproduces the statistics of $\iscripthat$, $\iscript$ gives more information than
$\iscripthat$. This is because when a measurement of $\iscript$ produces a result in $\Delta\in\fscript$ the instrument $\iscript$ updates the state of the system to the new state $\iscript (\Delta )\rho/\rmtr\sqbrac{\iscript (\Delta )\rho}$ when $\rmtr\sqbrac{\iscript (\Delta )\rho}\ne 0$ \cite{bgl95,hz12,nc00}.

If $\iscript$ is an instrument on $(\Omega _\iscript ,\fscript )$ and $v\colon\Omega _\iscript\times\gscript\to\sqbrac{0,1}$ is a stochastic kernel, then we shall show that
\begin{equation*}
(V\tbullet\iscript )(\Delta )=\int v(x,\Delta )\iscript (dx)
\end{equation*}
is an instrument with outcome space $(\Omega ,\gscript )$ called a \textit{coarse-graining} of $\iscript$. To show this we have that
$(V\tbullet\iscript )(\Delta )$ is countably additive on $\gscript$ and
\begin{equation*}
(V\tbullet\iscript )(\Omega )=\int v(x,\Omega )\iscript (dx)=\int _{\Omega _\iscript}\iscript (dx)=\iscript (\Omega _\iscript )
\end{equation*}
so $(V\tbullet\iscript )(\Omega )$ is a channel. Moreover, if $\rho\in\sscript (H)$ we obtain
\begin{align*}
\rmtr\sqbrac{(V\tbullet\iscript )(\Delta )\rho}&=\rmtr\sqbrac{\int v(x,\Delta )\iscript (dx)(\rho )}=\int v(x,\Delta )\rmtr\sqbrac{\iscript (dx)\rho}\\
   &\le\int\rmtr\sqbrac{\iscript (dx)\rho}=\rmtr\sqbrac{\iscript (\Omega _\iscript )\rho}=\rmtr (\rho )
\end{align*}
It follows that $V\tbullet\iscript$ is an instrument. It is easy to check that instruments form a convex set and that $\iscript\mapsto V\tbullet\iscript$ is affine.

\begin{thm}    % Theorem 2.6
\label{thm26}
{\rm{(a)}}\enspace $(V\tbullet\iscript )^\wedge =V\tbullet\iscripthat$.
{\rm{(b)}}\enspace For instruments $\iscript ,\jscript$ we have that $\Phi _\rho ^\jscript =V(\Phi _\rho ^\iscript )$ for all $\rho\in\sscript (H)$ if and only if
$\jscript =V\tbullet\iscripthat$.
{\rm{(c)}}\enspace If $\jscript =V\tbullet\iscript$, then $\Phi _\rho ^\jscript =V(\Phi _\rho ^\iscript )$ for all $\rho\in\sscript (H)$.
\end{thm}
\begin{proof}
(a)\enspace For all $\rho\in\sscript (H)$ we obtain
\begin{align*}
\rmtr\sqbrac{\rho (V\tbullet\iscripthat )(\Delta )}&=\rmtr\sqbrac{\rho\int v(x,\Delta )\iscripthat (dx)}=\int v(x,\Delta )\rmtr\sqbrac{\rho\iscripthat (dx)}\\
   &=\int v(x,\Delta )\rmtr\sqbrac{\iscript (dx)(\rho )}=\rmtr\sqbrac{\int v(x,\Delta )\iscript (dx)(\rho )}\\
   &=\rmtr\sqbrac{(V\tbullet\iscript )(\Delta )(\rho )}=\rmtr\sqbrac{\rho (V\tbullet\iscript )^\wedge (\Delta )}
\end{align*}
It follows that $(V\tbullet\iscript )^\wedge =V\tbullet\iscripthat$.
(b)\enspace If $\Phi _\rho ^\jscript =V(\Phi _\rho ^\iscript )$, then for all $\rho\in\sscript (H)$ we have that
\begin{align*}
\rmtr\sqbrac{\rho\jscripthat (\Delta )}&=\rmtr\sqbrac{\jscript (\rho )(\Delta )}=\Phi _\rho ^\jscript (\Delta )=V(\Phi _\rho ^\iscript )(\Delta )
   =\int v(x,\Delta )\Phi _\rho ^\iscript (dx)\\
   &=\int v(x,\Delta )\rmtr\sqbrac{\iscript (\rho )(dx)}=\int v(x,\Delta )\rmtr\sqbrac{\rho\iscripthat (dx)}\\
   &\rmtr\sqbrac{\rho\int v(x,\Delta )\iscripthat (dx)}=\rmtr\sqbrac{\rho V\tbullet\iscripthat (\Delta )}
\end{align*}
Therefore, $\jscripthat (\Delta )=V\tbullet\iscripthat (\Delta )$ for all $\Delta$ so $\jscripthat =V\tbullet\iscripthat$. Conversely, if $\jscripthat =V\tbullet\iscripthat$, then for all $\rho\in\sscript (H)$ we obtain
\begin{align*}
\Phi _\rho ^\jscript (\Delta )&=\rmtr\sqbrac{\rho\jscripthat (\Delta )}=\rmtr\sqbrac{\rho V\tbullet\iscripthat (\Delta )}
   =\rmtr\sqbrac{\rho\int v(x,\Delta )\iscripthat (dx)}\\
   &=\int v(x,\Delta )\rmtr\sqbrac{\rho\iscripthat (dx)}=\int v(x,\Delta )\rmtr\sqbrac{\iscript (\rho )(dx)}=\int v(x,\Delta )\Phi _\rho ^\iscript (dx)\\
   &=V(\Phi _\rho ^\iscript )(\Delta )
\end{align*}
Hence, $\Phi _\rho ^\jscript =V(\Phi _\rho ^\iscript )$.
(c)\enspace If $\jscript =V\tbullet\iscript$, then by (a) $\jscripthat =(V\tbullet\iscript )^\wedge =V\tbullet\iscripthat$. Applying (b) gives
$\Phi _\rho ^\jscript =V(\Phi _\rho ^\iscript )$ for all $\rho\in\sscript (H)$.
\end{proof}

The converse of Theorem~\ref{thm26}(c) does not hold. That is, if $\Phi _\rho ^\jscript =V(\Phi _\rho ^\iscript )$ for all $\rho\in\sscript (H)$, we need not have $\jscript =V\tbullet\iscript$. For example, let $V=I$ the identity map. Then $\Phi _\rho ^\jscript =\Phi _\rho ^\iscript$ for all $\rho\in\sscript (H)$. But there exist $\jscript\ne\iscript$ with $\Phi _\rho ^\jscript =\Phi _\rho ^\iscript$ for all $\rho\in\sscript (H)$ so $\jscript\ne I\tbullet\iscript =\iscript$. Applying Theorem~\ref{thm26}, we can consider the various special types of coarse-graining for instruments.

\section{Finite Observables}  % Section 3
In this section, we restrict our attention to finite observables. If $A$ is an observable with $\Omega _A=\brac{x_1,\ldots ,x_n}$, then $A$ is completely determined by
\begin{equation*}
\brac{A\paren{\brac{x_1}},A\paren{\brac{x_2}},\ldots ,A\paren{\brac{x_n}}}
\end{equation*}
We then define $A_x=A\paren{\brac{x}}$ and write $A=\brac{A_x\colon x\in\Omega _A}$. It follows that for all $\Delta\subseteq\Omega _A$ we have that $A(\Delta )=\sum\brac{A_x\colon x\in\Delta}$. Let $B=\brac{B_y\colon y\in\Omega _B}$ be another observable and let $V\colon\rmprob (\Omega _A)\to\rmprob (\Omega _B)$ be an affine map. We write $B=V\tbullet A$ if $\Phi _\rho ^B=V(\Phi _\rho ^A)$ for all $\rho\in\sscript (H)$. We then say that $B$ is a \textit{post-processing} of $A$ \cite{gud220,hz12}. Thus, post-processing is the same as coarse-graining for finite observables.

\begin{thm}    % Theorem 3.1
\label{thm31}
If $V\colon\rmprob (\Omega _A)\to\rmprob (\Omega _B)$ is affine, then $B=V\tbullet A$ if and only if
$B_y=\sum\limits _{x\in\Omega _A}\vtilde _{xy}A_x$ for all $y\in\Omega _B$ where $\vtilde _{xy}$ is the stochastic matrix corresponding to $V$.
\end{thm}
\begin{proof}
Suppose $V\colon\rmprob (\Omega _A)\to\rmprob (\Omega _B)$ is affine and $B=V\tbullet A$. By Theorem~\ref{thm17} $\vtilde$ is a stochastic matrix and for all $\rho\in\sscript (H)$ we obtain
\begin{align*}
\rmtr (\rho B_y)&=\Phi _\rho ^B(y)=V(\Phi _\rho ^A)(y)=\vtilde (\Phi _\rho ^A)(y)=\sum _{x\in\Omega _A}\vtilde _{xy}\Phi _\rho ^A(x)\\
   &=\sum _{x\in\Omega _A}\vtilde _{xy}\rmtr (\rho A_x)=\rmtr\sqbrac{\rho\sum _{x\in\Omega _A}\vtilde _{xy}A_x}
\end{align*}
It follows that $B_y=\sum\limits _{x\in\Omega _A}\vtilde _{xy}A_x$. Conversely, suppose $B_y=\sum _{x\in\Omega _A}\vtilde _{xy}A_x$ for all
$y\in\Omega _B$. Then for all $\rho\in\sscript (H)$ and $y\in\Omega _B$ we obtain
\begin{equation*}
\Phi _\rho ^B(y)=\rmtr (\rho B_y)=\rmtr\paren{\rho\sum _{x\in\Omega _A}\vtilde _{xy}A_x}=\sum _{x\in\Omega _A}\vtilde _{xy}\Phi _\rho ^A(x)
   =(V\Phi _\rho ^A)(y)
\end{equation*}
Hence, $B=V\tbullet A$.
\end{proof}

We can identify an observable with a set $A=\brac{A_{x_1},A_{x_2},\ldots ,A_{x_n}}\subseteq\escript (H)$ satisfying
$\sum\limits _{i=1}^nA_{x_i}=I$. We say that $A$ is \textit{rank}~1, \textit{sharp}, \textit{atomic}, respectively, if $A_{x_i}$ are rank~1, projections,
1-dimensional projections. If $A$ is sharp, it follows that $A_xA_y=A_yA_x=0$ for $x\ne y$ \cite{gud220,hz12}. If $A$ is atomic, there exists an
orthonormal basis $\brac{\phi _i}$ for $H$ such that $A_{x_i}=\ket{\phi _i}\bra{\phi _i}$, $i=1,2,\ldots ,n$. Notice that $A_x$ is rank~1 if and only if
$A_x=\lambda P$ where $0<\lambda\le 1$ and $P$ is a 1-dimensional projection.

If $A=\brac{A_x\colon x\in\Omega _A}$, $B=\brac{B_y\colon y\in\Omega _B}$ are observables on $H$, their \textit{sequential product} $A\circ B$ is the observable with outcome space $\Omega _A\times\Omega _B$ given by \cite{gg02,gud220}.
\begin{equation*}
(A\circ B)_{(x,y)}=A_x\circ B_y=A_x^{1/2}B_yA_x^{1/2}
\end{equation*}
We also define the observable $B$ \textit{conditioned} by the observable $A$ as
\begin{equation*}
(B\mid A)_y=\sum _{x\in\Omega _A}(A_x\circ B_y)
\end{equation*}
It can be shown that $(B\mid A)$ coexists with $A$ \cite{gud220}. If $\mu$ is a stochastic matrix of the appropriate size, then
\begin{align}                % equation (3.1)
\label{eq31}
(A\circ\mu\tbullet B)_{(x,y)}&=A_x\circ (\mu\tbullet B)_y=A\circ\sum _{z\in\Omega _B}\mu _{zy}B_z
   =\sum _{z\in\Omega _B}\mu_{zy}A_x\circ B_z\notag\\
   &=\sum _{z\in\Omega _B}\mu _{zy}A_x^{1/2}B_zA_x^{1/2}
\end{align}
and if $\nu$ is a stochastic matrix of the appropriate size, then
\begin{align}                % equation (3.2)
\label{eq32}
\sqbrac{(\nu\tbullet A)\circ B}_{(x,y)}&=(\nu\tbullet A)_x\circ B_y=\paren{\sum _{z\in\Omega _A}\nu _{zx}A_z}\circ B_y\notag\\
   &=\paren{\sum _{z\in\Omega _A}\nu _{zx}A_z}^{1/2}B_y\paren{\sum _{z\in\Omega _A}\nu _{zx}A_z}^{1/2}
\end{align}
Notice that \eqref{eq32} is much more complicated than \eqref{eq31}. If $A$ is sharp, then \eqref{eq31},\eqref{eq32} become
\begin{align}               % equation (3.3) & equation (3.4)   
\label{eq33}
(A\circ\mu\tbullet B)_{(x,y)}&=\sum _{z\in\Omega _B}\mu _{zy}A_xB_zA_x\\
\intertext{and}
\label{eq34}
\sqbrac{(\nu\tbullet A)\circ B}_{(x,y)}
   &=\paren{\sum _{z\in\Omega _A}\nu _{zx}^{1/2}A_z}B_y\paren{\sum _{z\in\Omega _A}\nu _{zx}^{1/2}A_z}\\
   &=\sum _{r,s\in\Omega _A}\nu _{rx}^{1/2}\nu _{sx}^{1/2}A_rB_yA_s\notag
\end{align}
If $A$ and $B$ are atomic with $A_x=\ket{\phi _x}\bra{\phi _x}$ and $B_y=\ket{\psi _y}\bra{\psi _y}$ then \eqref{eq31}, \eqref{eq32} become
\begin{align}               % equation (3.5) & equation (3.6)   
\label{eq35}
(A\circ\mu\tbullet B)_{(x,y)}&=\sqbrac{\sum _{z\in\Omega _B}\mu _{zy}\ab{\elbows{\phi _x,\psi _z}}^2}\ket{\phi _x}\bra{\phi _x}\\
\intertext{and}
\label{eq36}
\sqbrac{(\nu\tbullet A)\circ B}_{(x,y)}
   &=\sum _{r,s\in\Omega _A}\nu _{rx}^{1/2}\nu _{sx}^{1/2}\elbows{\phi _r,\psi _y}\elbows{\psi _y,\phi _s}\ket{\phi _r}\bra{\phi _s}\notag\\
   &=\ket{\sum _{r\in\Omega _A}\nu _{rx}^{1/2}\elbows{\phi _r,\psi _y}\phi _r}\bra{\sum _{r\in\Omega _A}\nu _{rx}^{1/2}\elbows{\phi _r,\psi _y}\phi _r}
\end{align}
Notice from \eqref{eq35} and \eqref{eq36} that both $A\circ\mu\tbullet B$ and $(\nu\tbullet A)\circ B$ are rank~1 observables. The next lemma shows that post-processing and conditioning interact in a regular way.

\begin{lem}    % Lemma 3.2
\label{lem32}
$(\mu\tbullet B\mid A)=\mu\tbullet (B\mid A)$
\end{lem}
\begin{proof}
The result follows because
\begin{align*}
(\mu\tbullet B\mid A)_z&=\sum _{x\in\Omega _A}\sqbrac{A_x\circ (\mu\tbullet B)_z}
   =\sum _{x\in\Omega _A}\paren{A_x\circ\sum _{y\in\Omega _B}\mu _{yz}B_y}\\
   &=\sum _{y\in\Omega _B}\mu _{yz}\sqbrac{\sum _{x\in\Omega _A}(A_x\circ B_y)}=\sqbrac{\mu\tbullet (B\mid A)}_z
\end{align*}
Hence, $(\mu\tbullet B\mid A)=\mu\tbullet (B\mid A)$.
\end{proof}

\begin{exam}{4}  % Example 4
This example illustrates the concepts of this section in terms of finite position and momentum observables. Let $H$ be a finite-dimensional Hilbert space with dimension $d$ and let $\brac{\phi _j\colon j=0,1,\ldots ,d-1}$ be an orthonormal basis for $H$. The \textit{finite Fourier transform} is the unitary operator on $H$ given by
\begin{equation*}
F=\frac{1}{\sqrt{d}\,}\sum _{j,k=0}^{d-1}e^{2\pi ijk/d}\ket{\phi _k}\bra{\phi _j}
\end{equation*}
where $i=\sqrt{-1}\,$ \cite{hz12}. Equivalently, $F$ is the operator satisfying
\begin{equation*}
F(\phi _k)=\frac{1}{\sqrt{d}\,}\sum _{j=0}^{d-1}e^{2\pi ijk/d}\phi _j
\end{equation*}
for all $k=0,1,\ldots ,d-1$. We call $Q=\brac{Q_j\colon j=0,1,\ldots ,d-1}$ where $Q_j=\ket{\phi _j}\bra{\phi _j}$ the \textit{finite position observable} and
$P=\brac{P_j\colon j=0,1,\ldots ,d-1}$ where $P_j=\ket{\psi _j}\bra{\psi _j}$ with $\psi _j=F\phi _j$ the \textit{finite momentum observable}. Notice that
$P_j=FQ_jF^*$, $j=0,1,\ldots ,d-1$ and $\Omega _Q=\Omega _P=\brac{0,1,\ldots ,d-1}$. We also see that $Q$ and $P$ are atomic observables. The observable $Q\circ P$ has effects
\begin{equation*}
(Q\circ P)_{(j,k)}=Q_j\circ P_k=Q_jP_kQ_j=\ab{\elbows{\phi _j,\psi _j}}^2\ket{\phi _j}\bra{\phi _j}=\tfrac{1}{d}\,Q_j
\end{equation*}
Thus, $Q\circ P$ is a rank~1 observable and $(P\mid Q)$ is the trivial observable
\begin{equation*}
(P\mid Q)_k=\sum _j(Q_j\circ P_k)=\tfrac{1}{d}\,I
\end{equation*}
for $k=0,1,\ldots ,d-1$. In a similar way, $(P\circ Q)_{(j,k)}=\tfrac{1}{d}\,P_j$ and $(Q\mid P)_k=\tfrac{1}{d}\,I$ for $j,k=0,1,\ldots ,d-1$. The distribution of $Q$ in the state $\rho\in\sscript (H)$ becomes
\begin{equation*}
\Phi _\rho ^Q(j)=\rmtr (\rho Q_j)=\elbows{\phi _j,\rho\phi _j}
\end{equation*}
for $j=0,1,\ldots ,d-1$.
\end{exam}   % 9.14.21 placed here to force paragraph, instead of \end{exam} below 

More interesting observables are obtained by post-processing. Let $\mu _{rj}$ be a stochastic matrix so that $\mu _{rj}\ge 0$ and
$\sum\limits _{j=0}^{d-1}\mu _{rj}=1$, for all $r,j=0,1,\ldots ,d-1$. Then the post-processing observable $\mu\tbullet Q$ satisfies
\begin{equation*}
(\mu\tbullet Q)_j=\sum _{r=0}^{d-1}\mu _{rj}Q_r=\sum _{r=0}^{d-1}\mu _{rj}\ket{\phi _r}\bra{\phi _r}
\end{equation*}
We see that the eigenvalues of $(\mu\tbullet Q)_j$ are $\mu _{rj}$, $r=0,1,\ldots ,d-1$ with corresponding eigenvectors $\phi _r$. The distribution of
$\mu\tbullet Q$ in the state $\rho\in\sscript (H)$ becomes
\begin{align*}
\Phi _\rho ^{\mu\tbullet Q}(j)&=\rmtr\sqbrac{\rho (\mu\tbullet Q)_j}=\rmtr\sqbrac{\rho\sum _{r=0}^{d-1}\mu _{rj}\ket{\phi _r}\bra{\phi _r}}\\
   &=\sum _{r=0}^{d-1}\mu _{rj}\elbows{\phi _r,\rho\phi _r}=\sum _{r=0}^{d-1}\mu _{rj}\Phi _\rho ^Q(r)
\end{align*}
The observable $(\mu\tbullet Q)\circ P$ satisfies
\begin{align}                % equation (3.7)
\label{eq37}
\sqbrac{(\mu\tbullet Q)\circ P}_{(j,k)}&=(\mu\tbullet Q)_j\circ P_k=(\mu\tbullet Q)_j^{1/2}P_k(\mu\tbullet Q_j)^{1/2}\notag\\
   &=\sum _{r=0}^{d-1}\mu _{rj}^{1/2}\ket{\phi _r}\bra{\phi _r}P_k\sum _{s=0}^{d-1}\mu _{sj}^{1/2}\ket{\phi _s}\bra{\phi _s}\notag\\
   &=\sum _{r,s=0}^{d-1}\mu _{rj}^{1/2}\mu _{sj}^{1/2}\elbows{\phi _r,\psi _k}\elbows{\psi _k,\phi _s}\ket{\phi _r}\bra{\phi _s}\notag\\
   &=\frac{1}{d}\sum _{r,s=0}^{d-1}\mu _{rj}^{1/2}\mu _{sj}^{1/2}e^{2\pi ik(s-r)}\ket{\phi _r}\bra{\phi _s}
\end{align}
Equation\eqref{eq37} also follows from \eqref{eq36}.\hfill\qedsymbol
%\end{exam}

\section{SIC Observables}  % Section 4
This section is more speculative than the previous ones and we do not come to many definite conclusions. A finite observable $A$ is
\textit{informationally complete} (IC) if $\rmtr (\rho _1A_x)=\rmtr (\rho _2A_x)$ for all $x\in\Omega _A$ implies that $\rho _1=\rho _2$. Equivalently, $A$ is informationally complete if $\Phi _{\rho _1}^A=\Phi _{\rho _2}^A$ implies that $\rho _1=\rho _2$. It can be shown that there exist IC observables for every finite dimensional Hilbert space $H$ \cite{hz12}. An Observable $A$ on a Hilbert space $H$ with $\dim H=d$ is \textit{symmetric} if \cite{hz12}:

\begin{list} {(S\arabic{cond})}{\usecounter{cond}
\setlength{\rightmargin}{\leftmargin}}
%(S1)
\item $\ab{\Omega _A}=d^2$,
%(S2)
\item $A$ has rank~1,
%(S3)
\item $\rmtr (A_x)=1/d$ for all $x\in\Omega _A$,
%(S4)
\item $\rmtr (A_xA_y)=1/d^2(d+1)$ for all $x\ne y\in\Omega _A$.
\end{list}
It can be shown that $d^2$ is the smallest cardinality for the outcome space of an IC observable \cite{hz12}. Also, $\rmtr (A_x)=1/d$ if $\rmtr (A_x)$ is constant for all $x\in\Omega _A$ and $\rmtr (A_xA_y)=1/d^2(d+1)$ for $x\ne y\in\Omega _A$ if $\rmtr (A_xA_y)$ is constant for $x\ne y$ \cite{hz12}. A symmetric IC observable is called a SIC \textit{observable}. An important unsolved problem is whether SIC observables exist for every finite dimensional Hilbert space \cite{hz12}. It is not even known whether high dimensional SIC observables exist. We would like to propose a possible method for attacking this problem. Unfortunately, we have not been able to complete this method and we leave this to future work.

Let $\dim H=d$ and let $A=\brac{\ket{\phi _x}\bra{\phi _x}\colon x\in\Omega _A}$, $B=\brac{\ket{\psi _y}\bra{\psi _y}\colon y\in\Omega _B}$ be atomic observables. For a $d\times d$ stochastic matrix $\mu$ we define the observable $C=(\nu\tbullet A)\circ B$. For example, $(\mu\tbullet Q)\circ P$ of Example~4 is such an observable. Letting $\eta _{xy}\in H$ be the vector given by
\begin{equation}                % equation (4.1)
\label{eq41}
\eta _{xy}=\sum _{r\in\Omega _A}\nu _{rx}^{1/2}\elbows{\phi _r,\psi _y}\phi _r
\end{equation}
We conclude from \eqref{eq36} that for all $(x,y)\in\Omega _C$ we have that
\begin{equation}                % equation (4.2)
\label{eq42}
C_{(x,y)}=\ket{\eta _{xy}}\bra{\eta _{xy}}
\end{equation}
It immediately follows that $C$ satisfies (S1) and (S2). We say that a stochastic matrix $\nu$ is \textit{doubly stochastic} if $\sum\limits _x\nu _{xy}=1$ for all $y$ \cite{hz12}. The bases $\brac{\phi _r}$, $\brac{\psi _y}$ are \textit{mutually unbiased bases} (MUB) if $\ab{\elbows{\phi _r,\psi _y}}^2=1/d$ for all $r,y=1,2,\ldots ,d$ \cite{gud220}. It is easy to show that there exist pairs of MUB for every finite dimension. In fact, the two bases in Example~4 are MUB.

\begin{thm}    % Theorem 4.1
\label{thm41}
{\rm{(a)}}\enspace If $\nu$ is doubly stochastic and $\brac{\phi _r}$, $\brac{\psi _y}$ are MUB, then $\rmtr\sqbrac{C_{(x,y)}}=1/d$ for all $x,y$.
{\rm{(b}}\enspace If $\rmtr\sqbrac{C_{(x,y)}}=1/d$ for all $x,y$ then $\nu$ is doubly stochastic.
\end{thm}
\begin{proof}
(a)\enspace Applying \eqref{eq41}, \eqref{eq42} we have that
\begin{align*}
\rmtr\sqbrac{C_{(x,y)}}&=\doubleab{\eta _{xy}}^2
   =\elbows{\sum _r\nu _{rx}^{1/2}\elbows{\phi _r,\psi _y}\phi _r,\sum _s\nu _{sx}^{1/2}\elbows{\phi _s,\psi _y}\psi _s}\\
   &=\sum _r\nu _{rx}\ab{\elbows{\phi _r,\psi _y}}^2
\end{align*}
for all $x,y$. If $\nu$ is doubly stochastic and $\brac{\phi _r}$, $\brac{\psi _y}$ are MUB we conclude that
\begin{equation*}
\rmtr\sqbrac{C_{(x,y)}}=\frac{1}{d}\sum _r\nu _{rx}=\frac{1}{d}
\end{equation*}
for all $x,y$.
(b)\enspace If $\rmtr\sqbrac{C_{(x,y)}}=1/d$ for all $x,y$, then by (a) we obtain
\begin{equation*}
\sum _r\nu _{rx}\ab{\elbows{\phi _r,\psi _y}}^2=\frac{1}{d}
\end{equation*}
for all $x,y$. Summing over $y$ gives $\sum\limits _r\nu _{rx}=1$ for all $x$, so $\nu$ is doubly stochastic.
\end{proof}

In Theorem~\ref{thm41}(b), if $\rmtr\sqbrac{C_{(x,y)}}=1/d$ for all $x,y$, then $\brac{\phi _r}$, $\brac{\psi _y}$ need not be MUB so the converse of Theorem~\ref{thm41}(a) does not hold. For example, suppose $\nu _{rx}=1/d$ for all $r,x$. Then $\rmtr\sqbrac{C_{(x,y)}}=1/d$ for all $x,y$ but
$\brac{\phi _r}$, $\brac{\psi _y}$ can be arbitrary bases. We conclude from Theorem~\ref{thm41}(a) that if $\nu$ is doubly stochastic and
$\brac{\phi _r}$, $\brac{\psi _y}$ are MUB, then Condition~(S3) holds.

\begin{lem}    % Lemma 4.2
\label{lem42}
{\rm{(a)}}\enspace Condition~(S4) holds if and only if $\elbows{\eta _{xy},\eta _{x'y'}}^2=1/d^2(d+1)$ for all $(x,y)\ne (x',y')$.
{\rm{(b)}}\enspace The observable $C$ is IC if and only if for $\rho _1,\rho _2\in\sscript (H)$ we have that
$\elbows{\rho _1\eta _{xy},\eta _{xy}}=\elbows{\rho _2\eta _{xy},\eta _{xy}}$ for all $x,y$ implies that $\rho _1=\rho _2$.
\end{lem}
\begin{proof}
(a)\enspace Applying \eqref{eq42} we have that
\begin{equation*}
\rmtr\sqbrac{C_{(x,y)}C_{(x',y')}}=\elbows{\eta _{xy},\eta _{x'y'}}^2
\end{equation*}
and the result follows.
(b)\enspace Applying \eqref{eq42} we have that
\begin{equation*}
\rmtr\sqbrac{\rho C_{(x,y)}}=\rmtr\paren{\rho\ket{\eta _{xy}}\bra{\eta _{xy}}}=\elbows{\rho\eta _{xy},\eta _{xy}}
\end{equation*}
and the result follows.
\end{proof}

Theorem~\ref{thm41} and Lemma~\ref{lem42} complete conditions under which $C$ become a SIC observable.

We now illustrate our SIC method in the qubit case $H=\complex ^2$. Let $\phi _1=(1,0)$, $\phi _2=(0,1)$ be the standard basis for $H$ and let $\brac{\psi _1,\psi _2}$ be a basis for $H$ such that $\brac{\phi _i}$, $\brac{\psi _j}$ are MUB. For example, we could use
\begin{align*}
\psi _1&=F(\phi _1)=\frac{1}{\sqrt{2}\,}\,(\phi _1+\phi _2)\\
\psi _2&=F(\phi _2)=\frac{1}{\sqrt{2}\,}\,(\phi _1-\phi _2)
\end{align*}
of Example~4. Define the atomic observables $A=\brac{\ket{\phi _1}\bra{\phi _1},\ket{\phi _2}\bra{\phi _2}}$, 
$B=\brac{\ket{\psi _1}\bra{\psi _1},\ket{\psi _2}\bra{\psi _2}}$ with $\Omega _A=\Omega _B=\brac{1,2}$. Let $\nu$ be the doubly stochastic matrix
\begin{equation*}
\nu=\begin{bmatrix}a&1-a\\1-a&a\end{bmatrix}
\end{equation*}
$0\le a\le 1$. Define the observable $C=(\nu\tbullet A)\circ B$ and the effects
\begin{equation*}
D=\begin{bmatrix}a^{1/2}&0\\0&(1-a)^{1/2}\end{bmatrix},\quad
   E=\begin{bmatrix}(1-a)^{1/2}&0\\0&a^{1/2}\end{bmatrix}
\end{equation*}
Letting $\eta _{xy}$, $x,y\in\brac{1,2}$, be the vectors defined by \eqref{eq41}, we have by \eqref{eq42} that $\eta _{11}=D\psi _1$,
$\eta _{12}+D\psi _2$, $\eta _{21}=E\psi _1$, $\eta _{22}=E\psi _2$.

We have that $C$ satisfies Conditions~(S1), (S2) and (S3). According to Theorem~\ref{thm41}(a), $C$ satisfies Condition~(S4) if and only if
$\elbows{\eta _{jk},\eta _{j'k'}}=1/12$ when $(j,k)\ne (j',k')$. Now
\begin{equation*}
\elbows{\eta _{11},\eta _{22}}^2=\elbows{D\psi _1,E\psi _2}^2=\elbows{ED\psi _1,\psi _2}^2=\elbows{a^{1/2}(1-a)^{1/2}\psi _1,\psi _2}^2=0
\end{equation*}
Hence, (S4) is not satisfied.

If $a=0$, $1$ or $1/2$, it is easy to check that $C$ is not IC. Unfortunately, even when $a\ne 0,1,1/2$, $C$ need not be IC. For example, let 
$\psi _1\tfrac{1}{\sqrt{2}\,}\,(\phi _1+\phi _2)$, $\psi _2=\tfrac{1}{\sqrt{2}\,}\,(\phi _1-\phi _2)$ as before. We then have the following result.

\begin{thm}    % Theorem 4.3
\label{thm43}
If $a\ne 0,1,1/2$ and $G$ is a $2\times 2$ self-adjoint matrix, $\rm\sqbrac{GC_{(j,k)}}=0$ for all $j,k=1,2$, if and only if 
$G=\begin{bmatrix}0&i\alpha\\-i\alpha&0\end{bmatrix}$ where $\alpha\in\real$.
\end{thm}
\begin{proof}
By \eqref{eq42}, $\rmtr\sqbrac{GC_{(j,k)}}=0$ if and only if
\begin{equation}                % equation (4.3)
\label{eq43}
\elbows{G\eta _{jk},\eta _{jk}}=0
\end{equation}
for all $j,k=1,2$. Letting $G=\begin{bmatrix}(1-a)^{1/2}&0\\0&a^{1/2}\end{bmatrix}$ we conclude that \eqref{eq43} holds if and only if
\begin{align}               % equation (4.4), equation (4.5), equation (4.6), equation (4.7)
\label{eq44}
\elbows{GD\psi _1,D\psi _1}&=\frac{1}{2}\sqbrac{G_{11}a+G_{22}(1-a)+2a^{1/2}(1-a)^{1/2}\rmre G_{12}}=0\\\noalign{\medskip}
\label{eq45}
\elbows{GE\psi _1,E\psi _1}&=\frac{1}{2}\sqbrac{G_{11}(1-a)+G_{22}a+2a^{1/2}(1-a)^{1/2}\rmre G_{12}}=0\\\noalign{\medskip}
\label{eq46}
\elbows{GD\psi _2,D\psi _2}&=\frac{1}{2}\sqbrac{G_{11}a+G_{22}(1-a)-2a^{1/2}(1-a)^{1/2}\rmre G_{12}}=0\\\noalign{\medskip}
\label{eq47}
\elbows{GE\psi _2,E\psi _2}&=\frac{1}{2}\sqbrac{G_{11}(1-a)+G_{22}a-2a^{1/2}(1-a)^{1/2}\rmre G_{12}}=0
\end{align}
Adding \eqref{eq44} and \eqref{eq46} gives $G_{11}a+G_{22}(1-a)=0$ and hence, $G_{11}=\tfrac{a-1}{a}\,G_{22}$. If $G_{22}\ne 0$, then
$(a-1)^2=a^2$ so $a=1/2$ which is a contradiction. Hence, $G_{22}=0$ and it follows that $G_{11}=\rmre G_{12}=0$. Hence, $G_{12}=i\alpha$,
$\alpha\in\real$ and the result follows. The converse is clear.
\end{proof}

\begin{cor}    % Corollary 4.4
\label{cor44}
If $\psi _1=\tfrac{1}{\sqrt{2}\,}\,(\phi _1+\phi _2)$, $\psi _2=\tfrac{1}{\sqrt{2}\,}\,(\phi _1-\phi _2)$, then $C$ is not IC.
\end{cor}
\begin{proof}
Define $\rho _1=\tfrac{1}{2}\begin{bmatrix}1&1\\0&1\end{bmatrix}$, $\rho _2=\begin{bmatrix}1/2&i\alpha\\-i\alpha&1/2\end{bmatrix}$ where
$0<\alpha<1/2$. Then $\rho _1\in\sscript (H)$ and it is easy to check that $\rho _2\in\sscript (H)$ by showing that the eigenvalues of $\rho _2$ are
$\tfrac{1}{2}\pm\alpha$. Since $\rho _2-\rho _1=\begin{bmatrix}0&i\alpha\\-i\alpha&0\end{bmatrix}$, it follows from Theorem~\ref{thm43} that
\begin{equation*}
\rmtr\sqbrac{\rho _2C_{(j,k)}}=\rmtr\sqbrac{\rho _1C_{(j,k)}}
\end{equation*}
for all $j,k=1,2$. But $\rho _2\ne\rho _1$ so $C$ is not IC.
\end{proof}

It is possible that for other $\brac{\psi _1,\psi _2}$ we obtain an IC observable $C$. It is also possible that for higher dimensional spaces we obtain SIC observables using this method. Even though $C$ is not IC, it satisfies two necessary (but not sufficient) conditions for IC \cite{hz12}(Prop~3.35). If $C$ is IC these conditions are:
(a)\enspace $C_{(j,k)}$ does not have both eigenvalues 0,1 and
(b)\enspace for all $j,k$ there exists $j',k'$ such that
\begin{equation*}
C_{(j,k)}C_{(j',k')}\ne C_{(j',k')}C_{(j,k)}
\end{equation*}
Indeed, (a) is clear and (b) follows from the fact that
\begin{equation*}
C_{(1,1)}C_{(1,2)}\ne C_{(1,2)}C_{(1,1)}\quad\hbox{and}\quad C_{(2,2)}C_{(2,1)}\ne C_{(2,1)}C_{(2,2)}
\end{equation*}


\begin{thebibliography}{99}
% ref 1
\bibitem{bgl95}P.~Busch, M.~Grabowski and P.~Lahti, \textit{Operational Quantum Physics}, Springer-Verlag, Berlin,
 1995.
% ref 2
\bibitem{fhl18}S.~Fillipov, T.~Heinosaari and L.~Lepp\"aj\"arvi, Simulability of observables in general probabilistic theories, \textit{Phys.~Rev.}\textbf{A97}, 062102 (2018).
% ref 3
\bibitem{gg02}S.~Gudder and R.~Greechie, Sequential products on effect algebras, \textit{Rep.~Math.~Phys.}
\textbf{49}, 87--111 (2002).
% ref 4
\bibitem{gud120}S.~Gudder, Parts and composites of quantum systems, arXiv:quant-ph 2009.07371 (2020).
% ref 5
\bibitem{gud220}--------, Combinations of quantum observables and instruments, arXiv:quant-ph 2010.08025 (2020)
% ref 6
\bibitem{hz12}T.~Heinosaari and M.~Ziman, \textit{The Mathematical Language of Quantum Theory}, Cambridge University Press, Cambridge, 2012.
% ref 7
\bibitem{hrsz09}T.~Heinosaari, D.~Reitzner, R.~Stano and M.~Ziman, Coexistence of quantum operations, \textit{J.~Phys.} \textbf{A42}, 365302 (2009).
% ref 8
\bibitem{lah03}P.~Lahti, Coexistence and joint measurability in quantum mechanics, \textit{Int.~J.~Theor.~Phys.} \textbf{42}, 893--906 (2003).
% ref 9
\bibitem{nc00}M.~Nielson and I.~Chuang, Quantum Computation and Quantum Information, Cambridge University Press, Cambridge, 2000.

\end{thebibliography}
\end{document}